\documentclass[a4paper,UKenglish,cleveref, autoref]{lipics-v2019}

\usepackage[utf8]{inputenc}
\usepackage{amsmath}
\usepackage{amssymb}
\usepackage{amsthm}
\usepackage{booktabs}
\usepackage{graphicx}
\usepackage{longtable}
\usepackage{mathtools}
\usepackage{pifont}
\usepackage{siunitx}
\usepackage{timetravel}
\usepackage{wasysym}

\newcommand{\cmark}{\ding{51}}%
\newcommand{\xmark}{\ding{55}}%

\setCounters{theorem}

\DeclareMathOperator{\dist}{dist}
\DeclareMathOperator{\acosh}{acosh}
\DeclareMathOperator{\poly}{poly}
\DeclareMathOperator{\iw}{iw}

\DeclareMathOperator{\pw}{pw}

\title{Solving Vertex Cover in Polynomial Time on Hyperbolic Random
  Graphs} 

\titlerunning{Solving Vertex Cover in Polynomial Time on Hyperbolic
 Random Graphs}

\author{Thomas Bläsius}{Hasso Plattner Institute, University of Potsdam\\{Potsdam, Germany}}{thomas.blaesius@hpi.de}{}{}

\author{Philipp Fischbeck}{Hasso Plattner Institute, University of Potsdam\\{Potsdam, Germany}}{philipp.fischbeck@hpi.de}{}{}

\author{Tobias Friedrich}{Hasso Plattner Institute, University of Potsdam\\{Potsdam, Germany}}{tobias.friedrich@hpi.de}{https://orcid.org/0000-0003-0076-6308}{}

\author{Maximilian Katzmann}{Hasso Plattner Institute, University of Potsdam\\{Potsdam, Germany}}{maximilian.katzmann@hpi.de}{}{}

\authorrunning{T. Bläsius, P. Fischbeck, T. Friedrich, M. Katzmann}

\Copyright{Thomas Bläsius, Philipp Fischbeck, Tobias Friedrich,
  Maximilian
  Katzmann}



\funding{This research was partially funded by the German Research
  Foundation (Deutsche\\Forschungsgemeinschaft, DFG) -- project number
  390859508.}

\nolinenumbers 

\hideLIPIcs  

\EventEditors{Christophe Paul and Markus Bl\"{a}ser}
\EventNoEds{2}
\EventLongTitle{37th International Symposium on Theoretical Aspects of Computer Science (STACS 2020)}
\EventShortTitle{STACS 2020}
\EventAcronym{STACS}
\EventYear{2020}
\EventDate{March 10--13, 2020}
\EventLocation{Montpellier, France}
\EventLogo{}
\SeriesVolume{154}
\ArticleNo{21}

\begin{document}
\maketitle

\begin{abstract}
  The \textsc{VertexCover} problem is proven to be computationally
  hard in different ways: It is NP-complete to find an optimal
  solution and even NP-hard to find an approximation with reasonable
  factors.  In contrast, recent experiments suggest that on many
  real-world networks the run time to solve \textsc{VertexCover} is
  way smaller than even the best known FPT-approaches can explain.
  Similarly, greedy algorithms deliver very good approximations to the
  optimal solution in practice.

  We link these observations to two properties that are observed in
  many real-world networks, namely a heterogeneous degree distribution
  and high clustering.  To formalize these properties and explain the
  observed behavior, we analyze how a branch-and-reduce algorithm
  performs on hyperbolic random graphs, which have become increasingly
  popular for modeling real-world networks.  In fact, we are able to
  show that the \textsc{VertexCover} problem on hyperbolic random
  graphs can be solved in polynomial time, with high probability.

  The proof relies on interesting structural properties of hyperbolic
  random graphs.  Since these predictions of the model are interesting
  in their own right, we conducted experiments on real-world networks
  showing that these properties are also observed in practice.  When
  utilizing the same structural properties in an adaptive greedy
  algorithm, further experiments suggest that, on real instances, this
  leads to better approximations than the standard greedy approach
  within reasonable time.
\end{abstract}

\newpage

\section{Introduction}
\label{sec:introduction}

\textsc{VertexCover} is a fundamental NP-complete graph problem.  For
a given undirected graph~$G$ on $n$ vertices the goal is to find the
smallest vertex subset~$S$, such that each edge in~$G$ is incident to
at least one vertex in $S$.  Since, by definition, there can be no
edge between two vertices outside of~$S$, these remaining vertices
form an independent set.  Therefore, one can easily derive a maximal
independent set from a minimal vertex cover and vice versa.

Due to its NP-completeness there is probably no polynomial time
algorithm for solving \textsc{VertexCover}.  The best known algorithm
for \textsc{IndependentSet} runs in $1.1996^n
\poly(n)$~\cite{xn-eamis-17}.  To analyze the complexity of
\textsc{VertexCover} on a finer scale, several parameterized solutions
have been proposed.  One can determine whether a graph $G$ has a
vertex cover of size~$k$ by applying a \emph{branch-and-reduce}
algorithm.  The idea is to build a search tree by recursively
considering two possible extensions of the current vertex cover
(\emph{branching}), until a vertex cover is found or the size of the
current cover exceeds~$k$.  Each branching step is followed by a
\emph{reduce} step in which \emph{reduction rules} are applied to make
the considered graph smaller.  This branch-and-reduce technique yields
a simple $\mathcal{O}(2^k \poly(n))$ algorithm, where the exponential
portion comes from the branching.  The best known FPT (fixed-parameter
tractable) algorithm runs in $\mathcal{O}(1.2738^k + kn)$
time~\cite{ckx-i-10}, and unless ETH (exponential time hypothesis)
fails, there can be no $2^{o(k)}\poly(n)$
algorithm~\cite{cj-oespa-03}.


While these FPT approaches promise relatively small running times if
the considered network has a small vertex cover, the cover is large
for many real-world networks.  Nevertheless, it was recently observed
that applying a branch-and-reduce technique on real instances is very
efficient~\cite{ai-bf-16}.  Some of the considered networks had
millions of vertices, yet an optimal solution (also containing
millions of vertices) was computed within seconds.  Most instances
were solved so quickly since the expensive branching was not necessary
at all.  In fact, the application of the reduction rules alone already
yielded an optimal solution.  Most notably, applying the
\emph{dominance reduction rule}, which eliminates vertices whose
neighborhood contains a vertex together with its neighborhood, reduces
the graph to a very small remainder on which the branching, if
necessary, can be done quickly.  We trace the effectiveness of the
dominance rule back to two properties that are often observed in
real-world networks: a \emph{heterogeneous degree distribution} (the
network contains many vertices of small degree and few vertices of
high degree) and \emph{high clustering} (the neighbors of a vertex are
likely to be neighbors themselves).

We formalize these key properties using \emph{hyperbolic random
  graphs} to analyze the performance of the dominance rule.
Introduced by Krioukov et al.~\cite{kpk-h-10}, hyperbolic random
graphs are obtained by randomly distributing nodes in the hyperbolic
plane and connecting any two that are geometrically close.  The
resulting graphs feature a power-law degree distribution and high
clustering~\cite{gpp-rhg-12, kpk-h-10} (the two desired properties)
which can be tuned using parameters of the model.  Additionally, the
generated networks have a small diameter~\cite{fk-dhrg-15}.  All of
these properties have been observed in many real-world networks such
as the internet, social networks, as well as biological networks like
protein-protein interaction networks.  Furthermore, Bogun{\'a},
Papadopoulos, and Krioukov showed that the internet can be embedded
into the hyperbolic plane such that routing packages between network
participants greedily works very well~\cite{bpk-sihm-10}, indicating
that this network naturally fits into the hyperbolic space.

By making use of the underlying geometry, we show that
\textsc{VertexCover} can be solved in polynomial time on hyperbolic
random graphs, with high probability.  This is done by showing that
even a single application of the dominance reduction rule reduces a
hyperbolic random graph to a remainder with small pathwidth on which
\textsc{VertexCover} can then be solved efficiently.  Our analysis
provides an explanation for why \textsc{VertexCover} can be solved
efficiently on practical instances.  We note that, while our analysis
makes use of the underlying hyperbolic geometry, the algorithm itself
is oblivious to it.  Besides the running time the model predicts
certain structural properties that also point us to an adapted greedy
algorithm that is still very efficient and achieves better
approximation ratios.  We conducted experiments indicating that these
predictions (concerning the structural properties and improved
approximation) actually match the real world for a significant
fraction of networks.

\section{Preliminaries}
\label{sec:preliminaries}

Let $G = (V, E)$ be an undirected graph.  We denote the number of
vertices in $G$ with $n$.  The \emph{neighborhood} of a vertex $v$ is
defined as $N(v) = \{ w \in V \mid \{ v, w \} \in E \}$ and the size
of the neighborhood, called the \emph{degree} of $v$, is denoted by
$\deg(v)$.  For a subset $S \subseteq V$, we use $G[S]$ to denote the
induced subgraph of $G$ obtained by removing all vertices in $V
\setminus S$.  Furthermore, we use the shorthand notation $G_{\le d}$
to denote $G[ \{ v \in V \mid \deg(v) \le d\}]$.

\subparagraph{The Hyperbolic Plane.}

After choosing a designated origin $O$ in the two-dimensional
hyperbolic plane, together with a reference ray starting at $O$, a
point $p$ is uniquely identified by its \emph{radius} $r(p)$, denoting
the hyperbolic distance to $O$, and its \emph{angle} (or \emph{angular
  coordinate})~$\varphi(p)$, denoting the angular distance between the
reference ray and the line through $p$ and $O$.  The hyperbolic
distance between two points $p$ and $q$ is given by
\begin{align}
  \dist(p, q) = \acosh(\cosh(r(p))\cosh(r(q)) - \sinh(r(p))\sinh(r(q))\cos(\Delta_\varphi(\varphi(p), \varphi(q)))), \notag
\end{align}
where $\cosh(x) = (e^x + e^{-x}) / 2$, $\sinh(x) = (e^x - e^{-x}) / 2$
(both growing as $e^x/2 \pm o(1)$), and $\Delta_\varphi(p, q) = \pi -
| \pi - | \varphi(p) - \varphi(q) ||$ denotes the angular distance
between $p$ and $q$.  If not stated otherwise, we assume that
computations on angles are performed modulo $2\pi$.

We use $B_p(r)$ to denote a disk of radius $r$ centered at $p$, i.e.,
the set of points with hyperbolic distance at most $r$ to $p$.  Such a
disk has an area of $2\pi(\cosh(r) - 1)$ and circumference
$2\pi\sinh(r)$.  Thus, the area and the circumference of a disk in the
hyperbolic plane grow exponentially with its radius.  In contrast,
this growth is polynomial in Euclidean space.  Therefore, representing
hyperbolic shapes in the Euclidean geometry results in a distortion.
In the \emph{native representation}, used in our figures, circles can
appear teardrop-shaped (see Figure~\ref{fig:domination-probability}).

\subparagraph{Hyperbolic Random Graphs.}

Hyperbolic random graphs are obtained by distributing $n$ points
uniformly at random within the disk $B_O(R)$ and connecting any two of
them if and only if their hyperbolic distance is at most $R$; see
Figure~\ref{fig:proof-idea}.  The disk radius $R$ (which matches the
connection threshold) is defined as $R = 2\log(8n/(\pi
\bar{\kappa}))$, where $\bar{\kappa}$ is a constant describing the
desired average degree of the generated network. The coordinates for
the vertices are drawn as follows. For vertex~$v$ the angular
coordinate, denoted by $\varphi(v)$, is drawn uniformly at random from
$[0, 2\pi]$ and the radius of $v$, denoted by $r(v)$, is sampled
according to the probability density function $\alpha\sinh(\alpha r) /
(\cosh(\alpha R) - 1)$ for $r \in [0, R]$ and $\alpha \in (1/2, 1)$.
Thus,
\begin{align}
  \label{eq:probability-density-function}
  f(r) = \frac{1}{2\pi} \frac{\alpha \sinh(\alpha r)}{\cosh(\alpha R) - 1} = \frac{\alpha}{2 \pi} e^{-\alpha(R - r)}(1 + \Theta(e^{-\alpha R} - e^{-2\alpha r})),
\end{align}
is their joint distribution function for $r \in [0, R]$.  For $r > R$,
$f(r) = 0$.  The constant $\alpha \in (1/2, 1)$ is used to tune the
power-law exponent $\beta = 2\alpha + 1$ of the degree distribution of
the generated network.  Note that we obtain power-law exponents $\beta
\in (2, 3)$.  Exponents outside of this range are atypical for
hyperbolic random graphs.  On the one hand, for $\beta < 2$ the
average degree of the generated networks is divergent.  On the other
hand, for $\beta > 3$ hyperbolic random graphs degenerate: They
decompose into smaller components, none having a size linear in $n$.
The obtained graphs have logarithmic tree width~\cite{bfk-hrg-16},
meaning the \textsc{VertexCover} problem can be solved efficiently in
that case.

The probability for a given vertex to lie in a certain area $A$ of the
disk is given by its probability measure $\mu(A) = \int_A f(r)
\mathrm{d}r$.  The hyperbolic distance between two vertices $u$ and
$v$ increases with increasing angular distance between them.  The
maximum angular distance such that they are still connected by an edge
is bounded by~\cite[Lemma~6]{gpp-rhg-12}
\begin{align}
  \label{eq:maximum-angular-distance}
  \theta(r(u), r(v)) &= \arccos\left( \frac{\cosh(r(u))\cosh(r(v)) - \cosh(R)}{\sinh(r(u))\sinh(r(v))} \right) \notag \\
                     &= 2e^{(R - r(u) - r(v))/2}(1 + \Theta(e^{R - r(u) - r(v)})).
\end{align}

\subparagraph{Interval Graphs and Circular Arc Graphs.}

In an interval graph each vertex $v$ is identified with an interval
on the real line and two vertices are adjacent if and only if their
intervals intersect.  The \emph{interval width} of an interval graph
$G$, denoted by $\iw(G)$, is its maximum clique size, i.e., the
maximum number of intervals that intersect in one point.  For any
graph the interval width is defined as the minimum interval width over
all of its interval supergraphs.  Circular arc graphs are a superclass
of interval graphs, where each vertex is identified with a subinterval
of the circle called \emph{circular arc} or simply \emph{arc}.  The
interval width of a circular arc graph $G$ is at most twice the size
of its maximum clique, since one obtains an interval supergraph of $G$
by mapping the circular arcs into the interval $[0, 2\pi]$ on the real
line and replacing all intervals that were split by this mapping with
the whole interval $[0, 2\pi]$.  Consequently, for any graph $G$, if
$k$ denotes the minimum over the maximum clique number of all circular
arc supergraphs $G'$ of $G$, then the interval width of $G$ is at most
$2k$.

\subparagraph{Treewidth and Pathwidth.}

A \emph{tree decomposition} of a graph $G$ is a tree $T$ where each
tree node represents a subset of the vertices of $G$ called
\emph{bag}, and the following requirements have to be satisfied: Each
vertex in $G$ is contained in at least one bag, all bags containing a
given vertex in $G$ form a connected subtree of $T$, and for each edge
in $G$, there exists a bag containing both endpoints.  The
\emph{width} of a tree decomposition is the size of its largest bag
minus one.  The \emph{treewidth} of $G$ is the minimum width over all
tree decompositions of $G$.  The \emph{path decomposition} of a graph
is defined analogously to the tree decomposition, with the constraint
that the tree has to be a path.  Additionally, as for the treewidth,
the \emph{pathwidth} of a graph $G$, denoted by $\pw(G)$, is the
minimum width over all path decompositions of $G$.  Clearly the
pathwidth is an upper bound on the treewidth.  It is known that for
any graph $G$ and any $k \ge 0$, the interval width of $G$ is at most
$k + 1$ if and only if its pathwidth is at most $k$~\cite[Theorem
7.14]{cfk-pa-15}.  Consequently, if $k'$ is the maximum clique size of
a circular arc supergraph of $G$, then $2k' - 1$ is an upper bound on
the pathwidth of $G$.

\subparagraph{Probabilities.}

Since we are analyzing a random graph model, our results are of
probabilistic nature.  To obtain meaningful statements, we show that
they hold \emph{with high probability} (for short \emph{whp.}), i.e.,
with probability $1 - \mathcal{O}(n^{-1})$.  The following Chernoff
bound is a useful tool for showing that certain events occur with high
probability.

\begin{theorem}[{Chernoff Bound~\cite[A.1]{dp-cmara-12}}]
  \label{thm:chernoff}
  Let $X_1, \dots, X_n$ be independent random variables with $X_i \in
  \{0, 1\}$ and let $X$ be their sum.  Let $f(n) = \Omega(\log(n))$.
  If $f(n)$ is an upper bound for $\mathbb{E}[X]$, then for each
  constant $c$ there exists a constant $c'$ such that $X \le c'f(n)$
  holds with probability $1 - \mathcal{O}(n^{-c})$.
\end{theorem}

\section{Vertex Cover on Hyperbolic Random Graphs}

\begin{figure}
  \centering
  \includegraphics{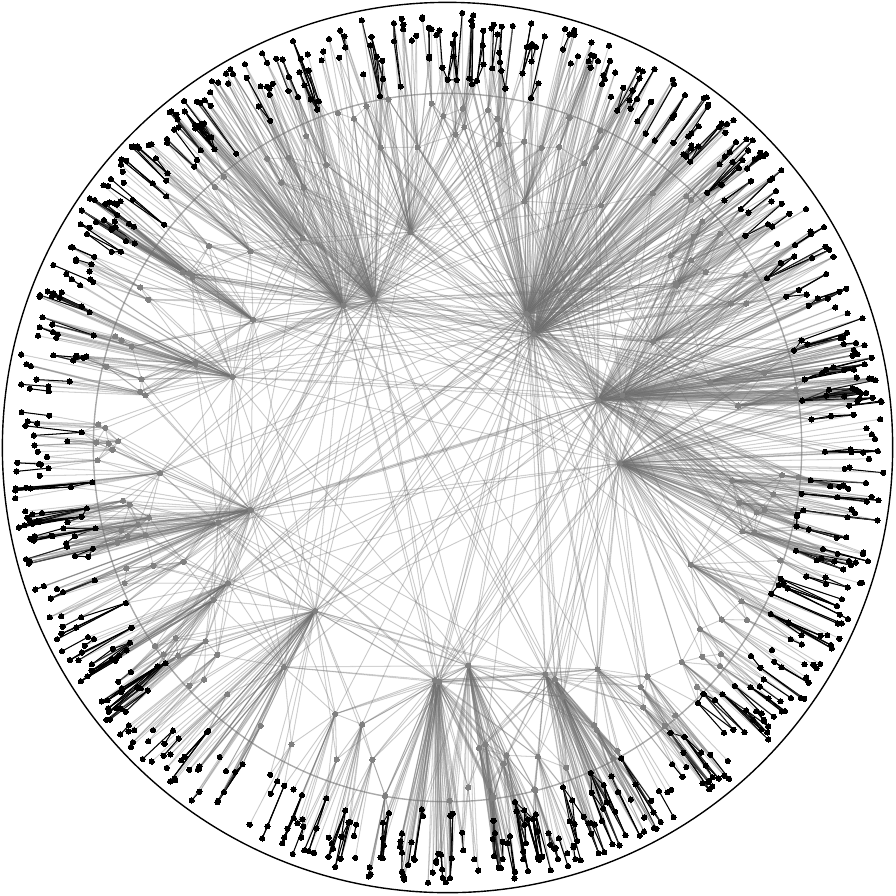}
  \caption{A hyperbolic random graph with $979$ nodes, average degree
    $8.3$, and a power-law exponent of $2.5$. With high probability,
    the gray vertices and edges are removed by the dominance reduction
    rule. Additionally, the remaining subgraph in the outer band
    (consisting of the black vertices and edges) has a small path
    width, with high probability.}
  \label{fig:proof-idea}
\end{figure}

Reduction rules are often applied as a preprocessing step, before
using a brute force search or branching in a search tree.  They
simplify the input by removing parts that are easy to solve.  For
example, an isolated vertex does not cover any edges and can thus
never be part of a minimum vertex cover.  Consequently, in a
preprocessing step all isolated vertices can be removed, which leads
to a reduced input size without impeding the search for a minimum.

The dominance reduction rule was previously defined for the
\textsc{IndependentSet} problem~\cite{fgk-mcaaea-09}, and later used
for \textsc{VertexCover} in the experiments by Akiba and
Iwata~\cite{ai-bf-16}.  Formally, vertex $u$ \emph{dominates} a
neighbor $v \in N(u)$ if $(N(v) \setminus \{u\}) \subseteq N(u)$,
i.e., all neighbors of $v$ are also neighbors of $u$.  We say $u$ is
\emph{dominant} if it dominates at least one vertex.  The dominance
rule states that $u$ can be added to the vertex cover (and afterwards
removed from the graph), without impeding the search for a minimum
vertex cover.  To see that this is correct, assume that $u$
dominates~$v$ and let $S$ be a minimum vertex cover that does not
contain $u$.  Since $S$ has to cover all edges, it contains all
neighbors of~$u$.  These neighbors include $v$ and all of $v$'s
neighbors, since $u$ dominates $v$.  Therefore, removing $v$ from $S$
leaves only the edge $\{u, v\}$ uncovered which can be fixed by adding
$u$ instead.  The resulting vertex cover has the same size as $S$.
When searching for a minimum vertex cover of $G$, it is thus safe to
assume that $u$ is part of the solution and to reduce the search to
$G[V \setminus \{u\}]$.

In the remainder of this section, we study the effectiveness of the
dominance reduction rule on hyperbolic random graphs and conclude that
\textsc{VertexCover} can be solved efficiently on these graphs.  Our
results are summarized in the following main theorem.

\wormhole{thm-vertex-cover-poly}
\begin{theorem}
  \label{thm:vertex-cover-poly}
  Let $G$ be a hyperbolic random graph on $n$ vertices.  Then the
  \textsc{VertexCover} problem on $G$ can be solved in $\poly(n)$
  time, with high probability.
\end{theorem}

The proof of Theorem~\ref{thm:vertex-cover-poly} consists of two parts
that make use of the underlying hyperbolic geometry.  In the first
part, we show that applying the dominance reduction rule once removes
all vertices in the inner part of the hyperbolic disk with high
probability, as depicted in Figure~\ref{fig:proof-idea}.  We note that
this is independent of the order in which the reduction rule is
applied, as dominant vertices remain dominant after removing other
dominant vertices.  In the second part, we consider the induced
subgraph containing the remaining vertices near the boundary of the
disk (black vertices in Figure~\ref{fig:proof-idea}).  We prove that
this subgraph has a small pathwidth, by showing that there is a
circular arc supergraph with a small interval width.  Consequently, a
tree decomposition of this subgraph can be computed efficiently.
Finally, we obtain a polynomial time algorithm for
\textsc{VertexCover} by first applying the reduction rules and
afterwards solving \textsc{VertexCover} on the remaining subgraph
using dynamic programming on the tree decomposition of small width.

\subsection{Dominance on Hyperbolic Random Graphs}
\label{sec:hyperbolic_domination}

\begin{figure}
  \centering
   \includegraphics{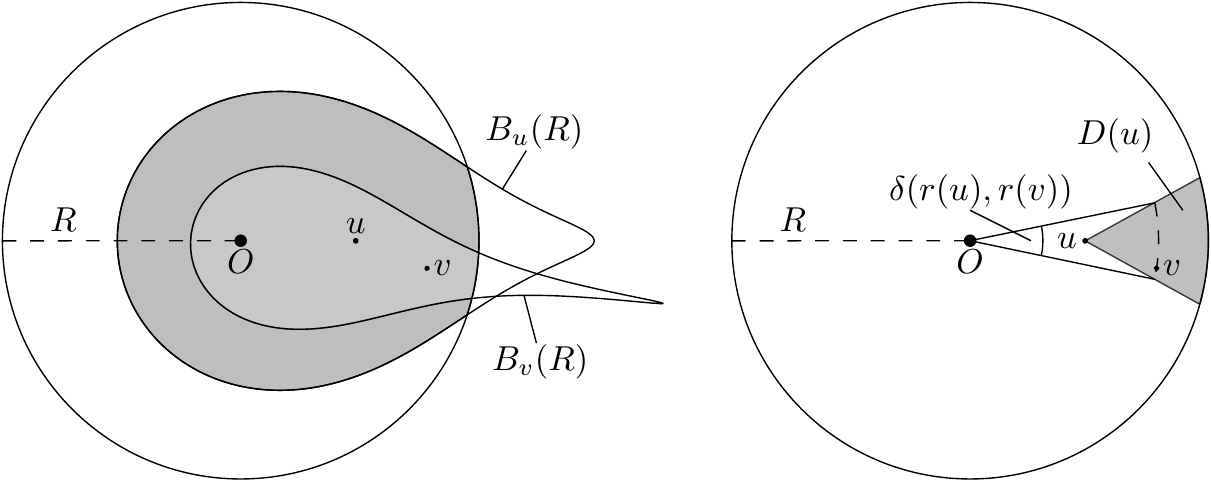}
   \caption{Left: Vertex $u$ dominates vertex $v$, as $B_v(R) \cap
     B_O(R)$ (light gray) is completely contained in $B_u(R) \cap
     B_O(R)$ (gray).  Right: All vertices that lie in $D(u)$ are
     dominated by $u$.}
  \label{fig:domination-probability}
\end{figure}

Recall that a hyperbolic random graph is obtained by distributing $n$
vertices in a hyperbolic disk $B_O(R)$ and that any two are connected
if their distance is at most $R$.  Consequently, one can imagine the
neighborhood of a vertex $u$ as another disk $B_u(R)$.  Vertex $u$
dominates another vertex $v$ if its neighborhood disk completely
contains that of $v$ (both constrained to $B_O(R)$), as depicted in
Figure~\ref{fig:domination-probability}~left.  We define the
\emph{dominance area} $D(u)$ of $u$ to be the area containing all such
vertices $v$.  That is, $D(u) = \{ p \in B_O(R) \mid B_p(R) \cap
B_O(R) \subseteq B_u(R) \cap B_O(R) \}$.  The result is illustrated in
Figure~\ref{fig:domination-probability}~right.  We note that it is
sufficient for a vertex $v$ to lie in $D(u)$ in order to be dominated
by $u$, however, it is not necessary.

Given the radius $r(u)$ of vertex $u$ we can now compute a lower bound
on the probability that $u$ dominates another vertex, i.e., the
probability that at least one vertex lies in $D(u)$, by determining
the measure $\mu(D(u))$.  To this end, we first define $\delta(r(u),
r(v))$ to be the maximum angular distance between two nodes $u$ and
$v$ such that $v$ lies in $D(u)$.

\begin{lemma}
  \label{lem:domination-angle}
  Let $u, v$ be vertices with $r(u) \le r(v)$.  Then, $v \in D(u)$ if
  $\Delta_\varphi(u, v)$ is at most
  \begin{displaymath}
    \delta(r(u), r(v)) = 2(e^{-r(u) / 2} - e^{-r(v) / 2}) + \Theta(e^{-3/2 r(u)}) - \Theta(e^{-3/2 r(v)}).
  \end{displaymath}
\end{lemma}
\begin{proof}
  Without loss of generality we assume that $\varphi(u) = 0$.  For now
  assume that $\varphi(v) = \varphi(u)$.  Since $r(v) \ge r(u)$ we
  know that the intersections of the boundaries of $B_v(R)$ with
  $B_O(R)$ lie between those of $B_u(R)$ with $B_O(R)$, as is depicted
  in Figure~\ref{fig:domination-angle}.  Now let $i_u$ denote one of
  these intersections for $B_u(R)$ and $B_O(R)$, and let $i_v$ denote
  the intersection for $B_v(R)$ and $B_O(R)$ that is on the same side
  of the ray through $O$ and $u$ as $i_u$.  It is easy to see that the
  maximum angular distance between $u$ and $v$ such that $B_v(R) \cap
  B_O(R)$ is contained within $B_u(R) \cap B_O(R)$ is given by the
  angular distance between $i_u$ and $i_v$.  Therefore, $v$ lies in
  the dominance area of $u$ if $\Delta_\varphi(u, v) \le
  \Delta_\varphi(i_u, i_v)$.

  Recall that $\theta(r(p), r(q))$ denotes the maximum angular
  distance such that $\dist(p, q) \le R$, as defined in
  Equation~\eqref{eq:maximum-angular-distance}.  Since $i_u$ and $i_v$
  have radius $R$ and hyperbolic distance $R$ to $u$ and $v$,
  respectively, we know that their angular coordinates are
  $\theta(r(u), R)$ and $\theta(r(v), R)$, respectively.
  Consequently, the angular distance between $i_u$ and $i_v$ is given
  by
  \begin{align}
    \delta(r(u), r(v)) &= \theta(r(u), R) - \theta(r(v), R) \notag \\
                       &= 2(e^{-r(u) / 2} - e^{-r(v) / 2}) + \Theta(e^{-3/2 r(u)}) - \Theta(e^{-3/2 r(v)}) \notag. \qedhere
  \end{align}
\end{proof}

\begin{figure}
  \centering
  \includegraphics{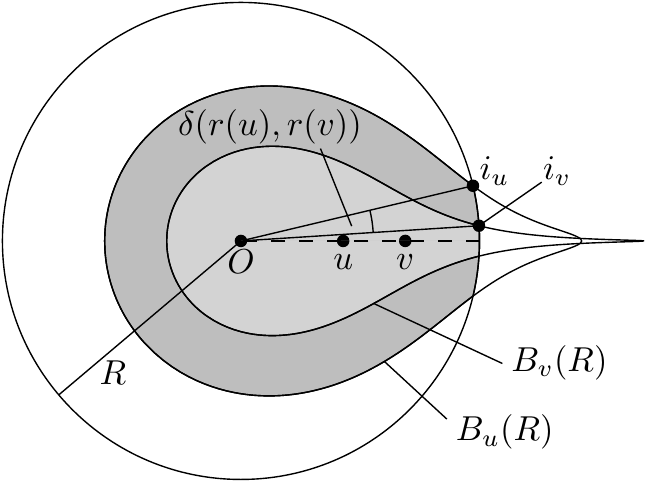}
  \caption{Vertex $u$ dominates vertex $v$, with $r(u) \le r(v)$, if
    $\Delta_\varphi(u, v) \le \Delta_\varphi(i_u, i_v)$.}
  \label{fig:domination-angle}
\end{figure}
  
Using Lemma~\ref{lem:domination-angle} we can now compute the
probability for a given vertex to lie in the dominance area of $u$.
We note that this probability grows roughly like $2/\pi e^{-r(u)/2}$,
which is a constant fraction of the measure of the neighborhood disk
of $u$ which grows as $\alpha / (\alpha - 1/2) \cdot 2 / \pi
e^{-r(u)/2}$~\cite[Lemma 3.2]{gpp-rhg-12}.  Consequently, the expected
number of nodes that $u$ dominates is a constant fraction of the
expected number of its neighbors.

\begin{lemma}
  \label{lem:domination-probability}
  Let $u$ be a node with radius $r(u) \ge R/2$.  The probability for a
  given node to lie in $D(u)$ is given by
  \begin{align}
    \mu(D(u)) &= \frac{2}{\pi} e^{-r(u) / 2} (1 - \Theta(e^{-\alpha(R - r(u))})) \pm \mathcal{O}(1/n) \notag. 
  \end{align}
\end{lemma}
\begin{proof}
  The probability for a given vertex $v$ to lie in $D(u)$ is obtained
  by integrating the probability density (given by
  Equation~\eqref{eq:probability-density-function}) over $D(u)$.
  \begin{align}
    \mu(D(u)) &= 2 \int_{r(u)}^R \int_{0}^{\delta(r(u), r)} f(r) \, \mathrm{d}\varphi \, \mathrm{d}r \notag \\
              &= 2 \int_{r(u)}^R \left( 2(e^{-r(u) / 2} - e^{-r / 2}) + \Theta(e^{-3/2 r(u)}) - \Theta(e^{-3/2 r}) \right) \notag \\
              &\hphantom{= 2 \int_{r(u)}^R} \cdot \frac{\alpha}{2\pi} e^{-\alpha(R - r)} (1 + \Theta(e^{-\alpha R} - e^{-2\alpha r})) \, \mathrm{d}r \notag 
  \end{align}
  Since $r(u) \ge R/2$ and $r \in [r(u), R]$ we have $\Theta(e^{-3/2
    r(u)}) - \Theta(e^{-3/2 r}) = \pm \mathcal{O}(e^{-3/4 R})$ and $(1
  + \Theta(e^{-\alpha R} - e^{-2 \alpha r})) = (1 + \Theta(e^{- \alpha
    R}))$.  Due to the linearity of integration, constant factors
  within the integrand can be moved out of the integral, which yields
  \begin{align}
    \mu(D(u)) &= \frac{\alpha}{\pi} e^{-\alpha R} (1 + \Theta(e^{-\alpha R})) \int_{r(u)}^R \left( 2(e^{-r(u) / 2} - e^{-r / 2}) \pm \mathcal{O}(e^{-3/4 R}) \right) \cdot e^{\alpha r} \, \mathrm{d}r \notag \\
              &= \frac{2 \alpha}{\pi} e^{-r(u) / 2} e^{-\alpha R} (1 + \Theta(e^{-\alpha R})) \int_{r(u)}^R e^{\alpha r} \mathrm{d}r \notag \\
              &\quad- \frac{2 \alpha}{\pi} e^{-\alpha R} (1 + \Theta(e^{-\alpha R})) \int_{r(u)}^R e^{(\alpha - 1/2)r} \mathrm{d}r \pm \mathcal{O} \left(e^{-(3/4 + \alpha) R} \int_{r(u)}^R e^{\alpha r} \mathrm{d}r \right). \notag
  \end{align}
  The remaining integrals can be computed easily and we obtain
  \begin{align}
    \mu(D(u)) &= \frac{2}{\pi} e^{-r(u) / 2}  (1 + \Theta(e^{-\alpha R})) (1 - e^{-\alpha (R - r(u))}) \notag \\
              &\quad- \frac{2 \alpha}{(\alpha - 1/2)\pi} e^{- R/2} (1 + \Theta(e^{-\alpha R})) (1 - e^{-(\alpha - 1/2)(R - r(u))}) \notag \\
              &\quad \pm \mathcal{O} \left(e^{-3/4 R} (1 - e^{-\alpha(R - r(u))}) \right). \notag
  \end{align}
  As $e^{-R/2} = \Theta(n^{-1})$ and $e^{-3/4 R} = \Theta(n^{-3/2})$,
  simplifying the error terms yields the claim.
\end{proof}

The following lemma shows that, with high probability, all vertices
that are not too close to the boundary of the disk dominate at least
one vertex.

\begin{lemma}
  \label{lem:whp-domination-near-center}
  Let $G$ be a hyperbolic random graph with average degree
  $\bar{\kappa}$.  Then there is a constant $c > 4/\bar{\kappa}$, such
  that all vertices $u$ with $r(u) \le \rho = R - 2\log\log(n^c)$ are
  dominant, with high probability.
\end{lemma}
\begin{proof}
  Vertex $u$ is dominant if at least one vertex lies in $D(u)$.  To
  show this for any $u$ with $r(u) \le \rho$, it suffices to show it
  for $r(u) = \rho$, since $D(u)$ increases with decreasing radius.
  To determine the probability that at least one vertex lies in
  $D(u)$, we use Lemma~\ref{lem:domination-probability} and obtain
  \begin{align}
    \mu(D(u)) &= \frac{2}{\pi} e^{-\rho / 2}(1 - \Theta(e^{-\alpha (R - \rho)})) \pm \mathcal{O}(1/n) \notag \\
              &= \frac{2}{\pi} e^{-R/2 + \log\log(n^c)} (1 - \Theta(e^{-2 \alpha \log\log(n^c)})) \pm \mathcal{O}(1/n). \notag
  \end{align}
  By substituting $R = 2\log(8n/ (\pi \bar{\kappa}))$, we obtain
  $\mu(D(u)) = \bar{\kappa}/(4n) (c \log(n) (1 - o(1)) \pm
  \mathcal{O}(1))$.
  The probability of at least one node falling into $D(u)$ is now
  given by
  \begin{align}
    \Pr[\{ v \in D(u) \} \neq \emptyset] = 1 - (1 - \mu(D(u)))^n \ge 1 - e^{-n \mu(D(u))} = 1 - \Theta(n^{- c \bar{\kappa} / 4 (1 - o(1))}). \notag
  \end{align}
  Consequently, for large enough $n$ we can choose $c >
  4/\bar{\kappa}$ such that the probability of a vertex at radius
  $\rho$ being dominant is at least $1 - \Theta(n^{-2})$, allowing us
  to apply union bound.
\end{proof}

\begin{corollary}
  \label{col:dominance-removes-inner-disk} 
  Let $G$ be a hyperbolic random graph and $c > 4/\bar{\kappa}$.  With
  high probability, all vertices with radius at most $\rho = R -
  2\log\log(n^c)$ are removed by the dominance rule.
\end{corollary}

By Corollary~\ref{col:dominance-removes-inner-disk} the dominance rule
removes all vertices of radius at most $\rho$.  Consequently, all
remaining vertices have radius at least $\rho$.  We refer to this part
of the disk as \emph{outer band}.  More precisely, the outer band is
defined as $B_O(R) \setminus B_O(\rho)$.  It remains to show that the
pathwidth of the subgraph induced by the vertices in the outer band is
small.

\subsection{Pathwidth in the Outer Band}
\label{sec:pathwidth-outer-band}

In the following, we use $G_r = G[\{v \in V\} \mid r(v) \ge r]$ to
denote the induced subgraph of $G$ that contains all vertices with
radius at least $r$.  To show that the pathwidth of $G_\rho$ (the
induced subgraph in the outer band) is small, we first show that there
is a circular arc supergraph $G_\rho^S$ of $G_\rho$ with a small
maximum clique.  We use $G^S$ to denote a circular arc supergraph of a
hyperbolic random graph $G$, which is obtained by assigning each
vertex $v$ an angular interval $I_v$ on the circle, such that the
intervals of two adjacent vertices intersect.  More precisely, for a
vertex~$v$, we set $I_v = [\varphi(v) - \theta(r(v), r(v)), \varphi(v)
+ \theta(r(v), r(v))]$.  Intuitively, this means that the interval of
a vertex contains a superset of all its neighbors that have a larger
radius, as can be seen in
Figure~\ref{fig:interval-representation}~left.  The following lemma
shows that $G^S$ is actually a supergraph of~$G$.

\begin{figure}
  \centering
  \includegraphics{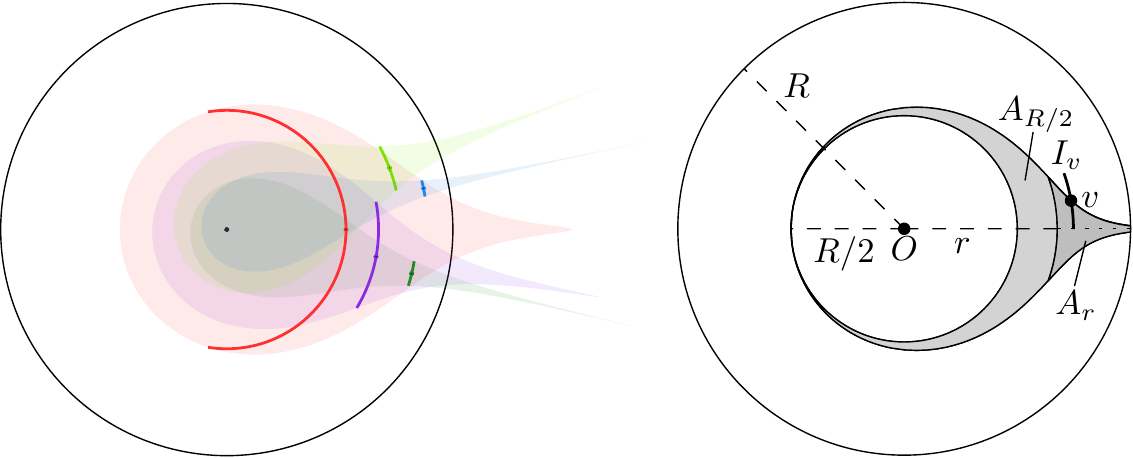}
  \caption{Left: The circular arcs representing the neighborhood of a
    vertex. For vertex $v$ the area containing the whole neighborhood
    of $v$, as well as the circular arc $I_v$ are drawn in the same
    color. Right: The area that contains the vertices whose arcs
    intersect angle $0$.  Area $A_{r}$ contains all such vertices with
    radius at least $r$.  Vertex $v$ lies on the boundary of $A_{r}$
    and its interval $I_v$ extends to $0$.}
  \label{fig:interval-representation}
\end{figure}

\begin{lemma}
  Let $G = (V, E)$ be a hyperbolic random graph.  Then $G^S$ is a
  supergraph of $G$.
\end{lemma}
\begin{proof}
  Let $\{ u, v \} \in E$ be any edge in $G$.  To show that $G^S$ is a
  supergraph of $G$ we need to show that $u$ and $v$ are also adjacent
  in $G^S$, i.e., $I_{u} \cap I_{v} \neq \emptyset$.  Without loss of
  generality assume $r(u) \le r(v)$.  Since $u$ and $v$ are adjacent
  in $G$, the hyperbolic distance between them is at most $R$.  It
  follows, that their angular distance $\Delta_{\varphi}(u, v)$ is
  bounded by $\theta(r(u), r(v))$.  Since $\theta(r(u), r(v)) \le
  \theta(r(u), r(u))$ for $r(u) \le r(v)$, we have
  $\Delta_{\varphi}(u, v) \le \theta(r(u), r(u))$.  As $I_{u}$ extends
  by $\theta(r(u), r(u))$ from $\varphi(u)$ in both directions, it
  follows that $\varphi(v) \in I_{u}$.
\end{proof}

It is easy to see that, after removing a vertex from $G$ and $G^S$,
$G^S$ is still a supergraph of $G$.  Consequently, $G_\rho^S$ is a
supergraph of $G_\rho$.  It remains to show that $G_\rho^S$ has a
small maximum clique number, which is given by the maximum number of
arcs that intersect at any angle.  To this end, we first compute the
number of arcs that intersect a given angle which we set to $0$
without loss of generality.  Let $A_r$ denote the area of the disk
containing all vertices $v$ with radius $r(v) \ge r$ whose interval
$I_v$ intersects $0$, as illustrated in
Figure~\ref{fig:interval-representation}~right.  The following lemma
describes the probability for a given vertex to lie in $A_r$.


\begin{lemma}
 \label{lem:expected-intersecting-arcs} 
 Let $G$ be a hyperbolic random graph and let $r \ge R/2$.  The
 probability for a given vertex to lie in $A_{r}$ is bounded by
 \begin{align}
   \mu(A_{r}) &\le \frac{2\alpha}{(1 - \alpha)\pi} e^{-(\alpha - 1/2)R - (1 - \alpha)r}  \cdot \left( 1 + \Theta(e^{-\alpha R} + e^{-(2r - R)} - e^{-(1 - \alpha)(R - r)}) \right). \notag
 \end{align}
\end{lemma}
\begin{proof}
  We obtain the measure of $A_r$ by integrating the probability
  density function over~$A_r$.  Due to the definition of $I_v$ we can
  conclude that $A_r$ includes all vertices $v$ with radius $r(v) \ge
  r$ whose angular distance to $0$ is at most $\theta(r(v), r(v))$,
  defined in Equation~(\ref{eq:maximum-angular-distance}).  We obtain,
  \begin{align}
    \mu(A_r) &= \int_r^R 2 \int_0^{\theta(x, x)} f(x) \, \mathrm{d}\varphi \, \mathrm{d}x \notag \\
             &= 2 \int_r^R 2e^{(R - 2x) / 2}(1 \pm \Theta(e^{R - 2x})) \cdot \frac{\alpha}{2 \pi} e^{-\alpha(R - x)} (1 + \Theta(e^{-\alpha R} - e^{-2\alpha x})) \, \mathrm{d}x. \notag
  \end{align}
  As before, we can conclude that $(1 + \Theta(e^{-\alpha R} - e^{-2
    \alpha r})) = (1 + \Theta(e^{-\alpha R}))$, since $r \ge R/2$.  By
  moving constant factors out of the integral, the expression can be
  simplified to
  \begin{align}
    \mu(A_r) &\le \frac{2 \alpha}{\pi} e^{-(\alpha - 1/2)R} (1 + \Theta(e^{-\alpha R})) \int_r^R e^{-(1 - \alpha)x}(1 + \Theta(e^{R - 2x})) \, \mathrm{d}x \notag .
  \end{align}
  We split the sum in the integral and deal with the two resulting
  integrals separately.
  \begin{align}
    \mu(A_r) &\le \frac{2 \alpha}{\pi} e^{-(\alpha - 1/2)R} (1 + \Theta(e^{-\alpha R})) \left( \int_r^R e^{-(1 - \alpha)x} \, \mathrm{d}x + \Theta \left( \int_r^R e^{-(1 - \alpha)x + R - 2x} \, \mathrm{d}x \right) \right) \notag \\
             &= \frac{2 \alpha}{\pi} e^{-(\alpha - 1/2)R} (1 + \Theta(e^{-\alpha R})) \notag \\
             &\qquad \cdot \Bigg( \frac{1}{1 - \alpha} e^{-(1 - \alpha)r}(1 - e^{-(1 - \alpha)(R - r)}) + \Theta \left( e^R e^{-(3 - \alpha)r}(1 - e^{-(3 - \alpha)(R - r)}) \right) \Bigg). \notag
  \end{align}
  By placing $1/(1 - \alpha)e^{-(1 - \alpha)r}$ outside of the
  brackets we obtain
  \begin{align}
    \mu(A_r) &\le \frac{2 \alpha}{(1 - \alpha)\pi} e^{-(\alpha - 1/2)R - (1 - \alpha)r} (1 + \Theta(e^{-\alpha R})) \notag \\
             &\qquad \cdot \Bigg( (1 - e^{-(1 - \alpha)(R - r)}) + \Theta \left( e^{R - 2r}(1 - e^{-(3 - \alpha)(R - r)}) \right) \Bigg). \notag
  \end{align}
  Simplifying the remaining error terms then yields the claim.
\end{proof}

We can now bound the maximum clique number in $G_\rho^S$ and thus its
interval width $\iw(G_\rho^S)$.

\begin{theorem}
  \label{thm:maximum-clique-number}
  Let $G$ be a hyperbolic random graph and $r \ge R / 2$.  Then there
  exists a constant $c$ such that, whp., $\iw(G_r^S) =
  \mathcal{O}(\log (n))$ if $r \ge R - \frac{1}{(1 - \alpha)}\log\log
  (n^c)$, and otherwise
  \begin{align}
    \iw(G_r^S) &\le \frac{4 \alpha}{(1 - \alpha)\pi} ne^{-(\alpha - 1/2)R - (1 - \alpha)r} \left( 1 + \Theta(e^{-\alpha R} + e^{-(2r - R)} - e^{-(1 - \alpha)(R - r)}) \right). \notag
  \end{align}
\end{theorem}
\begin{proof}
  We start by determining the expected number of arcs that intersect
  at a given angle, which can be done by computing the expected number
  of vertices in $A_{r}$, using
  Lemma~\ref{lem:expected-intersecting-arcs}:
  \begin{align}
    \mathbb{E}[|\{ v \in A_{r}\}|] &\le \frac{2\alpha}{(1 - \alpha)\pi} n e^{-(\alpha - 1/2)R - (1 - \alpha)r} ( 1 + \Theta(e^{-\alpha R} + e^{-(2r - R)} - e^{-(1 - \alpha)(R - r)}) ). \notag
  \end{align}
  It remains to show that this bound holds with high probability at
  every angle.  To this end, we make use of a Chernoff bound
  (Theorem~\ref{thm:chernoff}), by first showing that the bound on
  $\mathbb{E}[|\{ v \in A_{r}\}|]$ is $\Omega(\log(n))$.  We start
  with the case where $r < R - \frac{1}{1 - \alpha}\log\log(n^c)$.
  \begin{align}
    \mathbb{E}[|\{ v \in A_{r}\}|] &< \frac{2\alpha}{(1 - \alpha)\pi} n e^{-(\alpha - 1/2)R - (1 - \alpha)(R - 1/(1 - \alpha)\log\log(n^c))} \notag \\
                                   &\qquad \cdot \Big( 1 + \Theta(e^{-\alpha R} + e^{-(2(R - 1/(1 - \alpha)\log\log(n^c)) - R)} \notag \\
                                   & \hphantom{\qquad \cdot \Big(} - e^{-(1 - \alpha)(R - (R - 1/(1 - \alpha)\log\log(n^c)))}) \Big) \notag \\
                                   &= \frac{2\alpha}{(1 - \alpha)\pi} n e^{-R/2 + \log\log(n^c))} \notag \\
                                   &\qquad \cdot \Big( 1 + \Theta(e^{-\alpha R} + e^{-(R - 2/(1 - \alpha)\log\log(n^c))} - e^{-\log\log(n^c)}) \Big) \notag
  \end{align}
  Substituting $R = 2\log(8n/(\pi \bar{\kappa}))$ we obtain
  \begin{align}
                                       \mathbb{E}[|\{ v \in A_{r}\}|] &< \frac{\alpha \bar{\kappa} c}{4(1 - \alpha)} \log(n) (1 + o(1)). \notag
  \end{align}
  Thus, for all radii smaller than
  $R - \frac{1}{(1 - \alpha)}\log\log(n^c)$, the resulting upper bound
  is lower bounded by $\Omega(\log(n))$, which lets us apply
  Theorem~\ref{thm:chernoff}.  Moreover, as
  $\mathbb{E}[|\{ v \in A_{r}\}|]$ decreases with increasing $r$,
  $\mathcal O(\log(n))$ is a pessimistic but valid upper bound for the
  case $r \ge R - \frac{1}{(1 - \alpha)}\log\log(n^c)$.  Thus, we can
  also apply Theorem~\ref{thm:chernoff} to this case, using the
  $\mathcal O(\log(n))$~bound.

  By Theorem~\ref{thm:chernoff}, we can choose $c$ such that in both
  cases the bound holds with probability $1 - \mathcal{O}(n^{-c'})$
  for any $c'$ at a given angle.  In order to see that it holds at
  every angle, note that it suffices to show that it holds at all arc
  endings as the number of intersecting arcs does not change in
  between arc endings.  Since there are exactly $2n$ arc endings, we
  can apply union bound and obtain that the bound holds with
  probability $1 - \mathcal{O}(n^{-c' + 1})$ for any~$c'$ at every
  angle.  Since our bound on $\mathbb{E}[|\{ v \in A_r \}|]$ is an
  upper bound on the maximum clique size of $G_r^S$, the interval
  width of $G_r^S$ is at most twice as large, as argued in
  Section~\ref{sec:preliminaries}.
\end{proof}

Since the interval width of a circular arc supergraph of $G$ is an
upper bound on the pathwidth of $G$~\cite[Theorem 7.14]{cfk-pa-15} and
since $\rho \ge R - 1/(1 - \alpha)\log\log(n^c)$ for $\alpha \in (1/2,
1)$, we immediately obtain the following corollary.

\begin{corollary}
  \label{col:pathwidth}
  Let G be a hyperbolic random graph and let $G_\rho$ be the subgraph
  obtained by removing all vertices with radius at most $\rho = R -
  2\log\log(n^c)$.  Then, $\pw(G_\rho) = \mathcal{O}(\log(n))$.
\end{corollary}

We are now ready to prove our main theorem, which we
restate for the sake of readability.

\begin{backInTime}{thm-vertex-cover-poly}
  \begin{theorem}
    Let $G$ be a hyperbolic random graph on $n$ vertices.  Then the
    \textsc{VertexCover} problem in $G$ can be solved in $\poly(n)$
    time, with high probability.
  \end{theorem}
  \begin{proof}
    Consider the following algorithm that finds the minimum vertex
    cover of $G$.  We start with an empty vertex cover $S$.
    Initially, all dominant vertices are added to $S$, which is
    correct due to the dominance rule.  By
    Lemma~\ref{lem:whp-domination-near-center}, this includes all
    vertices of radius at most $\rho = R - 2\log\log(n^c)$, for some
    constant $c$, with high probability.  Obviously, finding all
    vertices that are dominant can be done in $\poly(n)$ time.  It
    remains to determine a vertex cover of $G_\rho$.  By
    Corollary~\ref{col:pathwidth}, the pathwidth of $G_\rho$ is
    $\mathcal{O}(\log(n))$, with high probability.  Since the
    pathwidth is an upper bound on the treewidth, we can find a tree
    decomposition of $G_\rho$ and solve the \textsc{VertexCover}
    problem in $G_\rho$ in $\poly(n)$ time~\cite[Theorems 7.18 and
    7.9]{cfk-pa-15}.
  \end{proof}
\end{backInTime}

Moreover, linking the radius of a vertex in
Theorem~\ref{thm:maximum-clique-number} with its expected degree leads
to the following corollary, which is interesting in its own right.  It
links the pathwidth to the degree $d$ in the graph $G_{\le d}$.
Recall that $G_{\le d}$ denotes the subgraph of $G$ induced by the
vertices of degree at most $d$.

\begin{corollary}
  Let $G$ be a hyperbolic random graph and let $d \le \sqrt{n}$.
  Then, with high probability, $\pw(G_{\le d}) = \mathcal{O}(d^{2 -
    2\alpha} + \log(n))$.
\end{corollary}
\begin{proof}
  Consider the radius $r = R - 2 \log(\varepsilon d)$ for some
  constant $\varepsilon > 0$, and the graph $G_{r}$ which is obtained
  by removing all vertices of radius at most $r$.  By substituting $R
  = 2\log(8n/(\pi \bar{\kappa}))$ and using~\cite[Lemma
  3.2]{gpp-rhg-12} we can compute the expected degree of a vertex
  with radius $r$ as
  \begin{align}
    \mathbb{E}[\deg(v) \mid r(v) = r] &= \frac{2 \alpha}{(\alpha - 1/2) \pi} n e^{-r/2}(1 \pm \mathcal{O}(e^{-(\alpha - 1/2)r})) = \frac{\alpha \bar{\kappa} \varepsilon}{4(\alpha - 1/2)} d (1 \pm o(1)). \notag
  \end{align}

  First assume that $d \ge \log(n)^{1/(2-2\alpha)}$.  We handle the
  other case later.
  Since $d \in \Omega(\log(n))$ we can choose $\varepsilon$ large
  enough to apply Theorem~\ref{thm:chernoff} and conclude that this
  holds with high probability. Furthermore, since a smaller radius
  implies a larger degree, we know that, with high probability, all
  nodes $v$ with radius at most $r$, have
  \begin{align}
    \deg(v) \ge \frac{\alpha \bar{\kappa} \varepsilon}{4(\alpha - 1/2)} d (1 \pm o(1)). \notag
  \end{align}
  For large enough $n$ we can choose $\varepsilon$ such that, with
  high probability, $G_r$ is a supergraph of $G_{\le d}$.  To prove
  the claim, it remains to bound the pathwidth of $G_r$.  If $r > R -
  1/(1 - \alpha)\log\log(n^c)$, we can apply the first part of
  Theorem~\ref{thm:maximum-clique-number} to obtain $\iw(G_r^S) =
  \mathcal{O}(\log(n))$.  Otherwise, we use part two to conclude that
  the interval width of $G_r$ is at most
  \begin{align}
    \iw(G_r^S) &\le \frac{4 \alpha}{(1 - \alpha)\pi} ne^{-(\alpha - 1/2)R - (1 - \alpha)r} \left( 1 + \Theta(e^{-\alpha R} + e^{-(2r - R)} - e^{-(1 - \alpha)(R - r)}) \right) \notag \\
               &= \frac{\alpha \bar{\kappa} \varepsilon^{2 - 2\alpha} }{(2 - 2\alpha)} d^{2 - 2\alpha} \left( 1 + \Theta(n^{-2\alpha} + ((\varepsilon d)^2/n)^2 - (\varepsilon d)^{-(2 - 2\alpha)}) \right) = \mathcal{O}(d^{2 - 2\alpha}). \notag
  \end{align}
  As argued in Section~\ref{sec:preliminaries} the interval width of a
  graph is an upper bound on the pathwidth.

  For $d < \log(n)^{1/(2-2\alpha)}$ (which we excluded above),
  consider $G_{\le d'}$ for $d' = \log(n)^{1/(2-2\alpha)} > d$.  As we
  already proved the corollary for $d'$, we obtain $\pw(G_{\le d'}) =
  \mathcal{O}(d'^{2 - 2\alpha} + \log(n)) = \mathcal{O}(\log(n))$.  As
  $G_{\le d}$ is a subgraph of $G_{\le d'}$, the same bound holds for
  $G_{\le d}$.
\end{proof}

\section{Discussion}
\label{sec:discussion}

Our results show that a heterogeneous degree distribution as well as
high clustering make the dominance rule very effective.  This matches
the behavior for real-world networks, which typically exhibit these
two properties.  However, our analysis actually makes more specific
predictions: (I) vertices with sufficiently high degree usually have
at least one neighbor they dominate and can thus safely be included in
the vertex cover; and (II) the graph remaining after deleting the high
degree vertices has simple structure, i.e., small pathwidth.

To see whether this matches the real world, we run experiments on
\num{59} networks from several network datasets~\cite{gephi-datasets,
  bm-p-06, k-k-13, lk-sd-14, ra-ndrigav-15}.  Although the focus of
this paper is the theoretical analysis on hyperbolic random graphs, we
briefly report on our experimental results; see Table~\ref{tab:data}
in Appendix~\ref{sec:experimental-data}.  Out of the \num{59} instances, we can
solve \textsc{VertexCover} for \num{47} networks in reasonable time.
We refer to these as \emph{easy}, while the remaining \num{12} are
called \emph{hard}.  Note that our theoretical analysis aims at
explaining why the easy instances are easy.

Recall from Lemma~\ref{lem:whp-domination-near-center} that all
vertices with radius at most $R - 2\log\log(n^{4/\bar{\kappa}})$
probably dominate, which corresponds to an expected degree of $\alpha
/ (\alpha - 1/2) \cdot \log n$.  For more than half of the \num{59}
networks, more than \SI{78}{\percent} of the vertices above this
degree were in fact dominant.  For more than a quarter of the
networks, more than \SI{96}{\percent} were dominant.  Restricted to
the \num{47} easy instances, these number increase to
\SI{82}{\percent} and \SI{99}{\percent}, respectively.

Experiments concerning the pathwidth of the resulting graph are much
more difficult, due to the lack of efficient tools.  Therefore, we
used the tool by Tamaki et al.~\cite{tosm-tmpt-17} to heuristically
compute upper bounds on the treewidth instead.  As in our analysis, we
only removed vertices that dominate in the original graph instead of
applying the reduction rule exhaustively.  On the resulting subgraphs,
the treewidth heuristic ran with a \SI{15}{min} timeout.  The
resulting treewidth is at most \num{50} for \SI{44}{\percent} of the
networks, at most \num{15} for \SI{34}{\percent}, and at most \num{5}
for \SI{25}{\percent}.  Restricted to easy instances, the values
increase to \SI{55}{\percent}, \SI{43}{\percent}, and
\SI{32}{\percent}, respectively.

Hyperbolic random graphs are of course an idealized representation of
real-world networks.  However, these experiments indicate that the
predictions derived from the model match the real world, at least for
a significant fraction of networks.

\subparagraph{Approximation.}

Concerning approximation algorithms for \textsc{VertexCover}, there is
a similar theory-practice gap as for exact solutions.  In theory,
there is a simple 2-approximation and the best known polynomial time
approximation reduces the factor to $2 -
\Theta(\log(n)^{-1/2})$~\cite{k-barvcp-09}.  However, it is NP-hard to
approximate \textsc{VertexCover} within a factor of
$1.3606$~\cite{ds-hamvc-05}, and presumably it is even NP-hard to
approximate within a factor of $2 - \varepsilon$ for all $\varepsilon
> 0$~\cite{kr-v-08}.  Moreover, the greedy strategy that iteratively
adds the vertex with maximum degree to the vertex cover and deletes
it, is only a $\log n$ approximation.  However, on scale-free networks
this strategy performs exceptionally well with approximation ratios
very close to~1~\cite{dgd-vccn-13}.

Our results for hyperbolic random graphs at least partially explain
this good approximation ratio.
Lemma~\ref{lem:whp-domination-near-center} states that, with high
probability, we do not make any mistake by taking all vertices below a
certain radius $\rho$, which corresponds to vertices of at least
logarithmic degree.  The same computation for larger values of $\rho$
does no longer give such strong guarantees.  However, it still gives
bounds on the probability for making a mistake.  In fact, this error
probability is sub-constant as long as the corresponding expected
degree is super-constant.

Although this is not a formal argument, it still explains to a degree
why greedy works so well on networks with a heterogeneous degree
distribution and high clustering.  Moreover, it indicates how the
greedy algorithm should be adapted to obtain better approximation
ratios: As the probability to make a mistake grows with growing radius
and thus with shrinking vertex degree, the majority of mistakes are
done when all vertices have already low degree.  However, for
hyperbolic random graphs, the subgraphs induced by vertices below a
certain constant degree decompose into small components for $n \to
\infty$.  It thus seems to be a good idea to run the greedy algorithm
only until all remaining vertices have low degree, say $k$.  The
remaining small connected components of maximum-degree $k$ can then be
solved with brute force in reasonable time.  In the following we call
the resulting algorithm \emph{$k$-adaptive greedy}.

We ran experiments on the \num{47} easy real networks mentioned above
(for the hard instances, we cannot measure approximation ratios).  For
these networks, we compare the normal greedy algorithm with 2- and
4-adaptive greedy.  Note that 2-adaptive greedy is special, as
\textsc{VertexCover} can be solved efficiently on graphs with maximum
degree~2 (no brute-forcing is necessary).  For 4-adaptive greedy, the
size of the largest connected component is relevant.

The median approximation ratio for greedy over all \num{47} networks
is \num{1.008}.  This goes down to \num{1.005} for 2-adaptive and to
\num{1.002} for 4-adaptive greedy.  Thus, the number of too many
selected vertices goes down by a factor of \num{1.6} and \num{4},
respectively.  As mentioned above, the size of the largest connected
component is relevant for 4-adaptive greedy.  For \SI{49}{\percent} of
the networks, this was below~\num{100} (which is still a reasonable
size for a brute-force algorithm).  Restricted to these networks,
normal greedy has a median approximation ratio of \num{1.004}, while
4-adaptive again improves by a factor of 4 to \num{1.001}.  Moreover,
the number of networks for which we actually obtain the optimal
solution increases from \num{4} to \num{7}.



\bibliography{hyperbolic_vertex_cover_arxiv}

\section{Experimental Data}
\label{sec:experimental-data}

Table~\ref{tab:data} (continuing on the next page) shows the raw data
of our experiments for which we reported aggregate values in the
discussion in Section~\ref{sec:discussion}.  The percentage of
dominant vertices among those with high degree (over
$\alpha / (\alpha - 1/2) \cdot \log n$) is rounded to whole
percentages.  The approximation ratios are rounded to three decimal
digits.  Treewidth $-1$ indicates that remaining graph after removing
all dominant vertices contained no edge.

\begin{longtable}{lcrrrrrr}
  \caption{The raw data of our experiments.  The columns
  are: \textbf{(network)} the network's name; \textbf{(easy)} whether
  or not we could compute an optimal solution; \textbf{(dom)}~the
  percentage of
  dominant vertices among the high-degree vertices;
  \textbf{(tw)}~an upper bound for the treewidth of the remaining graph
  after deleting dominant nodes; \textbf{(greedy)} the approximation
  ratio of
  greedy; \textbf{(2-ad)} the approximation
  ratio of 2-adaptive greedy; \textbf{(4-ad)} the approximation
  ratio of
  4-adaptive greedy; \textbf{(comp)} the size of the largest
  component that remains after the greedy phase of 4-adaptive
  greedy.}
  \label{tab:data}\\
  \toprule
  \textbf{network}            & \textbf{easy} & \textbf{dom}       & \textbf{tw} & \textbf{greedy} & \textbf{2-ad} & \textbf{4-ad} & \textbf{comp} \\
  \midrule\endhead
  \bottomrule \endfoot
  advogato                    & \cmark        & \SI{051}{\percent} & 314         & 1.011           & 1.009         & 1.005         & 863           \\
  airlines                    & \cmark        & \SI{028}{\percent} & 23          & 1.000           & 1.000         & 1.000         & 75            \\
  as-22july06                 & \cmark        & \SI{100}{\percent} & 3           & 1.002           & 1.001         & 1.001         & 46            \\
  as-caida20071105            & \cmark        & \SI{100}{\percent} & 3           & 1.002           & 1.001         & 1.000         & 35            \\
  as-skitter                  & \xmark        & \SI{047}{\percent} & 969794      &                 &               &               &               \\
  as20000102                  & \cmark        & \SI{100}{\percent} & 2           & 1.003           & 1.001         & 1.001         & 18            \\
  bio-CE-HT                   & \cmark        & \SI{100}{\percent} & 3           & 1.015           & 1.009         & 1.000         & 225           \\
  bio-CE-LC                   & \cmark        & \SI{100}{\percent} & 2           & 1.003           & 1.003         & 1.003         & 39            \\
  bio-DM-HT                   & \cmark        & \SI{050}{\percent} & 13          & 1.017           & 1.014         & 1.004         & 319           \\
  bio-yeast-protein-inter     & \cmark        & \SI{100}{\percent} & 4           & 1.013           & 1.006         & 1.002         & 147           \\
  bn-fly-drosophila-medulla-1 & \cmark        & \SI{072}{\percent} & 38          & 1.018           & 1.013         & 1.009         & 142           \\
  bn-mouse-kasthuri-graph-v4  & \cmark        & \SI{100}{\percent} & 1           & 1.006           & 1.000         & 1.000         & 12            \\
  ca-AstroPh                  & \cmark        & \SI{094}{\percent} & 6           & 1.003           & 1.002         & 1.001         & 123           \\
  ca-cit-HepPh                & \cmark        & \SI{084}{\percent} & 151         & 1.003           & 1.003         & 1.002         & 533           \\
  ca-CondMat                  & \cmark        & \SI{099}{\percent} & 4           & 1.003           & 1.002         & 1.001         & 53            \\
  ca-GrQc                     & \cmark        & \SI{099}{\percent} & 2           & 1.004           & 1.002         & 1.001         & 44            \\
  ca-HepTh                    & \cmark        & \SI{095}{\percent} & 13          & 1.005           & 1.004         & 1.001         & 174           \\
  cfinder-google              & \xmark        & \SI{066}{\percent} & 82          &                 &               &               &               \\
  cit-HepTh                   & \xmark        & \SI{013}{\percent} & 19737       &                 &               &               &               \\
  citeseer                    & \xmark        & \SI{046}{\percent} & 182372      &                 &               &               &               \\
  com-amazon                  & \cmark        & \SI{093}{\percent} & 2756        & 1.011           & 1.006         & 1.002         & 16209         \\
  com-dblp                    & \cmark        & \SI{100}{\percent} & 7           & 1.002           & 1.001         & 1.000         & 69            \\
  cpan-authors                & \cmark        & \SI{100}{\percent} & 2           & 1.009           & 1.009         & 1.009         & 17            \\
  digg-friends                & \cmark        & \SI{058}{\percent} & 1649        & 1.008           & 1.006         & 1.004         & 179           \\
  ego-facebook                & \cmark        & \SI{100}{\percent} & -1          & 1.000           & 1.000         & 1.000         & 3             \\
  ego-gplus                   & \cmark        & \SI{100}{\percent} & 1           & 1.000           & 1.000         & 1.000         & 5             \\
  email-Enron                 & \cmark        & \SI{085}{\percent} & 41          & 1.003           & 1.002         & 1.001         & 141           \\
  EuroSiS                     & \cmark        & \SI{056}{\percent} & 34          & 1.020           & 1.018         & 1.010         & 274           \\
  facebook-wosn-links         & \xmark        & \SI{027}{\percent} & 36694       &                 &               &               &               \\
  flixster                    & \xmark        & \SI{073}{\percent} & 122         &                 &               &               &               \\
  hyves                       & \cmark        & \SI{098}{\percent} & 1653        & 1.008           & 1.008         & 1.008         & 42            \\
  livemocha                   & \cmark        & \SI{004}{\percent} & 24380       & 1.017           & 1.013         & 1.006         & 25300         \\
  loc-brightkite-edges        & \cmark        & \SI{076}{\percent} & 619         & 1.014           & 1.009         & 1.004         & 4658          \\
\end{longtable}
\setcounter{table}{0}
\begin{longtable}{lcrrrrrr}
  \caption{The raw data of our experiments.  The columns
  are: \textbf{(network)} the network's name; \textbf{(easy)} whether
  or not we could compute an optimal solution; \textbf{(dom)}~the
  percentage of
  dominant vertices among the high-degree vertices;
  \textbf{(tw)}~an upper bound for the treewidth of the remaining graph
  after deleting dominant nodes; \textbf{(greedy)} the approximation
  ratio of
  greedy; \textbf{(2-ad)} the approximation
  ratio of 2-adaptive greedy; \textbf{(4-ad)} the approximation
  ratio of
  4-adaptive greedy; \textbf{(comp)} the size of the largest
  component that remains after the greedy phase of 4-adaptive
  greedy.}
  \\
  \toprule
  \textbf{network}            & \textbf{easy} & \textbf{dom}       & \textbf{tw} & \textbf{greedy} & \textbf{2-ad} & \textbf{4-ad} & \textbf{comp} \\
  \midrule\endhead
  \bottomrule \endfoot
  loc-gowalla-edges           & \xmark        & \SI{064}{\percent} & 3991        &                 &               &               &               \\
  moreno-names                & \cmark        & \SI{094}{\percent} & 3           & 1.006           & 1.004         & 1.002         & 34            \\
  moreno-propro               & \cmark        & \SI{100}{\percent} & 4           & 1.014           & 1.006         & 1.002         & 153           \\
  munmun-twitter-social       & \cmark        & \SI{057}{\percent} & 12          & 1.000           & 1.000         & 1.000         & 5             \\
  OClinks                     & \cmark        & \SI{036}{\percent} & 202         & 1.017           & 1.015         & 1.005         & 498           \\
  p2p-Gnutella04              & \cmark        & \SI{042}{\percent} & 1352        & 1.019           & 1.017         & 1.016         & 970           \\
  p2p-Gnutella05              & \cmark        & \SI{040}{\percent} & 1075        & 1.014           & 1.013         & 1.013         & 447           \\
  p2p-Gnutella06              & \cmark        & \SI{040}{\percent} & 1142        & 1.023           & 1.022         & 1.021         & 820           \\
  p2p-Gnutella08              & \cmark        & \SI{047}{\percent} & 414         & 1.008           & 1.008         & 1.008         & 45            \\
  p2p-Gnutella09              & \cmark        & \SI{047}{\percent} & 419         & 1.005           & 1.005         & 1.005         & 63            \\
  p2p-Gnutella24              & \cmark        & \SI{081}{\percent} & 525         & 1.006           & 1.005         & 1.005         & 70            \\
  p2p-Gnutella25              & \cmark        & \SI{079}{\percent} & 464         & 1.006           & 1.005         & 1.005         & 77            \\
  p2p-Gnutella30              & \cmark        & \SI{079}{\percent} & 604         & 1.005           & 1.005         & 1.004         & 62            \\
  p2p-Gnutella31              & \cmark        & \SI{080}{\percent} & 732         & 1.011           & 1.010         & 1.010         & 65            \\
  petster-carnivore           & \cmark        & \SI{079}{\percent} & 149312      & 1.008           & 1.007         & 1.004         & 9238          \\
  petster-friendship-cat      & \xmark        & \SI{012}{\percent} & 14929       &                 &               &               &               \\
  petster-friendship-dog      & \xmark        & \SI{015}{\percent} & 340634      &                 &               &               &               \\
  petster-friendship-hamster  & \xmark        & \SI{023}{\percent} & 135         &                 &               &               &               \\
  soc-Epinions1               & \cmark        & \SI{082}{\percent} & 238         & 1.006           & 1.003         & 1.001         & 228           \\
  US-Air                      & \cmark        & \SI{067}{\percent} & 4           & 1.013           & 1.000         & 1.000         & 23            \\
  web-Google                  & \xmark        & \SI{084}{\percent} & 103939      &                 &               &               &               \\
  wiki-Vote                   & \cmark        & \SI{044}{\percent} & 384         & 1.054           & 1.052         & 1.050         & 726           \\
  wordnet-words               & \cmark        & \SI{095}{\percent} & 28          & 1.004           & 1.003         & 1.002         & 59            \\
  YeastS                      & \cmark        & \SI{070}{\percent} & 39          & 1.013           & 1.012         & 1.005         & 244           \\
  youtube-links               & \cmark        & \SI{086}{\percent} & 1239        & 1.008           & 1.004         & 1.001         & 570           \\
  youtube-u-growth            & \xmark        & \SI{090}{\percent} & 59358       &                 &               &               &               \\
\end{longtable}

\end{document}